\documentclass[runningheads]{llncs}
\usepackage{booktabs}
\usepackage{multirow}
\usepackage{tikz}
\usetikzlibrary{patterns}
\usepackage{amsmath,amssymb}
\newtheorem{assumption}{Assumption}
\usepackage{graphicx}
\usepackage{enumitem}
\usepackage{appendix}

\usepackage{url}
\usepackage[hidelinks]{hyperref}

\usepackage{pgf, tikz}
\usetikzlibrary{arrows.meta, positioning, calc, fit, backgrounds, shapes.geometric, decorations.pathreplacing}

\usepackage[T1]{fontenc}
\newcommand{\name}{$\text{TEE}^{\text{BFT}}$}

\title{\texorpdfstring{\name{}: Pricing the Security of Proof of Cloud}{TEE-BFT: Pricing the Security of Proof of Cloud}}

\author{Alex Shamis\inst{1}\thanks{Authors listed in alphabetical order.} \and Matt Stephenson\inst{1} \and Linfeng Zhou\inst{1}}
\authorrunning{A. Shamis \and M. Stephenson \and L. Zhou} 

\institute{Subzero Labs}
\date{} 

\begin{document}
\maketitle

\begin{abstract}
  Blockchains face inherent limitations when communicating outside their own ecosystem. Trusted Execution Environments (TEEs) are a promising mitigation because they allow trusted brokers to interface with external systems. In this work, we develop a cost-of-collusion principal–agent model for compromising a TEE in a Data Center Execution Assurance design \cite{rezabek2025proofclouddatacenter}. Our model focuses on four core determinants of attack profitability: (i) a $K$‑of‑$n$ threshold for successful collusion, (ii) detection risk, (iii) per‑member sanctions $F_i$ conditional on being caught, and (iv) an extractable \emph{flow} reward $\omega$, which is distinguished from the total value (``stock'') of the system. 

\bigskip

  With these primitives we derive (a) a condition for rational collusion, (b) closed‑form deterrence thresholds showing how modifying parameters affects rational collusive, and (c) a design bound that guarantees non‑profitability of profitable collusion. Calibrations informed by time‑advantaged arbitrage \cite{fritsch2024timeadv} and estimates of security breach fallout demonstrate that protocols can plausibly secure on the order of \$1T in value against TEE compromise. We allude to this secure design as \name{}. 
\end{abstract}

\section{A Cost of Collusion Model of TEE Security}
\label{sec:principal-agent-model}
TEE security can be strengthened through the use of cloud providers \cite{rezabek2025proofclouddatacenter}, as well as multi-party-computation key splitting approaches under BFT assumptions e.g. \cite{gao2021teekap}. In this paper we model and attempt to estimate the plausible security available for using these methods.

Our model attends to TEE breaches with a principal–agent framework that captures the
incentives of trusted-execution-environment (TEE) providers in this setting. The goal is to quantify when collusion is not profitable for rational agents akin to e.g. \cite{becker1968crime}, given detection probabilities, penalties, and practical flow constraints on what attackers can seize during a feasible window.

Other related work on flow extraction and timing frictions include Flash Boys~2.0 which documented MEV and priority‑gas auctions, linking ordering rents to consensus‑layer risk~\cite{daian2019flashboys}. Fritsch–Mamageishvili–Silva–Livshits–Felten formalize \emph{time‑advantaged arbitrage}, showing under martingale‑like prices that optimal arbitrage waits until the end of the advantage window, with policy variants (e.g., auctions) reallocating a share of that flow~\cite{fritsch2024timeadv}. Our “stock vs.\ flow” calibration uses exchange turnover as a conservative proxy to bound short‑window extractable value; recent WFE and SIFMA statistics give global market‑cap and value‑traded orders of magnitude, and Budish’s economic‑limits argument is a conceptual antecedent to our “cheapest attack remains uneconomical” design condition~\cite{wfe-2024,sifma-2025,budish2018limits}.

Censorship dynamics in fraud‑proof protocols have been cast as explicit budgeted games. Berger–Felten–Mamageishvili–Sudakov derive challenge‑period lengths ensuring defender success as a function of move counts and attacker\slash defender budgets~\cite{berger2025econcensor}. Coalition‑proofness (CPNE) provides a canonical equilibrium notion for multi‑party deviation, justifying checks at threshold coalitions under monotone detection~\cite{bernheim1987cpne}.

Proof of Cloud~\cite{rezabek2025proof} have explored the ability to detect where a TEE is being hosted and examined the security of TEEs in cloud environments, but does not consider economic incentives directly. Systems like Enigma~\cite{zyskind2018enigma} and CCF (Confidential Consortium Framework)~\cite{howard2023confidential,russinovich2019ccf}, have proposed to leverage TEEs for privacy-preserving computations on blockchain networks with similar requirements for TEE security properties. However, to date there has been little attention to identifying thresholds at which it would be economically rational for agents to collude, and recent work on ``proof of cloud'' makes plausible our economic estimation of the security that may be available under.

Our paper's primary contribution is in modeling the multiplicative security available under such designs. Notably, our model does not assume that TEEs are secure against attacks--it is conservatively assumed that physical access is sufficient to extract the secrets of a TEE in the model. Instead we focus on the incentives of the service providers under breach and collusion detection risk.

\section{Model}\label{sec:model}
\subsection{Environment and threshold}\label{subsec:env}
There are \(n\ge2\) providers \(i\in\{1,\dots,n\}\).
A system event (e.g., a threshold decryption) occurs if at least \(K\) providers act in concert. We attend to the simple majority rule often preferred to ensure liveness
\begin{equation}
K(n)\;=\;\left\lfloor\frac{n}{2}\right\rfloor+1.
\label{eq:majority}
\end{equation}
\paragraph{Core Primitives: Rewards and Sanctions}
Then consider that provider \(i\) faces a per-member sanction scale \(F_i>0\) (legal, reputational, and/or balance-sheet exposure) if detected. For heterogeneity in $F$, we let \(F_{(1)}\ge\cdots\ge F_{(n)}\) denote the order statistics. 

The bounty that can be captured by successful colluders is $\omega$, which we treat as having equivalent units to $F$. We treat this as a fraction not greater than the total value secured by the system, \(V\ge0\). Intuitively, if a group could collude to discover the secret upcoming trades on the stock market, they would have to maintain that access for some length of time before they could capture full value of the stock market itself. This is affected by the observed flow rate (which is about 5\% per month for the US Stock market), the flow rate conditional on suspicious activity (e.g. the colluders have not been discovered, but traders have observed suspicious activity and are being more cautious and trading less), as well as time-governing aspects of the security breach itself (e.g. if cloud providers use different makes of TEE, they require any discovered vulnerabilities to overlap and not be patched),  That is, $V$ is the ``stock'' and $\omega \le V$ defines the ``capturable flow'' for the colluding group, with the proportional relationship governed by \(\beta\in(0,1]\).  Thus we characterize an extractable flow ``prize'' of
\begin{equation}
\omega\;=\;\beta\,V.
\label{eq:flow}
\end{equation}

\paragraph{Roadmap.} The primitives are a flow prize $\omega=\beta V$ and per-member sanction scales $\{F_i\}$. Two frictions govern feasibility: (i) a threshold $K$ of providers must act for the event to occur; and (ii) independent detection at rate $q$ per member induces both pre-coordination and execution risk. We first analyze the complete-information game and show that only the two corner symmetric profiles can be equilibria, with existence of the collusive corner governed by a simple odds-ratio condition. We then introduce dispersed information and obtain a unique symmetric cutoff equilibrium à la global games, which selects the collusive boundary as noise vanishes. Finally, we study heterogeneity in $F$ under a Zipf law to characterize how majority size $K$ shifts deterrence.

\subsection{Detection and success}\label{subsec:detection}
Let \(q\in(0,1)\) be the per-member independent detection probability. Following \cite{polinsky1979optimal,polinsky2000deterrence}, observe that for a coalition of size \(m\ge1\),
\begin{equation}
p(m)\;=\;1-(1-q)^m,\qquad
p_K:=p(K),\quad \tilde p:=p(K-1).
\label{eq:pofm}
\end{equation}
Note that we allow for $0<\tilde p \le p_K$ to account for the fact that detection probability may be lower (but not strictly zero) before a collusive group has formed. Given a belief \(\alpha\in[0,1]\) that each other provider joins, the probability that at least \(K-1\) of the other \(n{-}1\) providers join is
\begin{equation}
\pi_{n,K}(\alpha)\;=\;\Pr\!\big[\mathrm{Bin}(n{-}1,\alpha)\ge K{-}1\big].
\label{eq:success}
\end{equation}

\paragraph{Detection layers.} We model detection in two layers. First, pre-coordination exposes at least $K-1$ members (e.g., outreach, key-share solicitation), so an attempt incurs detection probability $\tilde p:=p(K-1)$. Second, if a size-$K$ coalition forms and executes, there is incremental exposure $p_K-\tilde p$. Hence the joiner's unconditional detection probability at belief $\alpha$ is $\bar p(\alpha)=\tilde p+\pi_{n,K}(\alpha)\big(p_K-\tilde p\big)$.

\subsection{Coalition composition and the binding type}\label{subsec:binding}
We assume the deviating coalition can choose its members. It is then optimal to select the \(K\) providers with the smallest sanction scales.

\begin{lemma}[Binding type]\label{lem:binding}
Fix \((n,K)\) and a profile \((F_i)_{i=1}^n\). The coalition’s participation is pinned down by the binding (highest-$F$) member among the chosen $K$, and $F_{\mathrm{eff}}$ enters individual incentives. If there exists a profitable coalition of size \(K\), there exists one consisting of the \(K\) lowest-\(F\) providers, with the marginal (binding) member having type
\begin{equation}
F_{\mathrm{eff}}(n,K)\;=\;F_{(n-K+1)}.
\label{eq:Feff-def}
\end{equation}
\end{lemma}

\begin{proof}[]
Replacing any coalition member \(j\) with a provider \(\ell\) such that \(F_\ell<F_j\) weakly raises each member's payoff: the prize share is unchanged while the expected sanction falls. Iterating this replacement yields a coalition of the \(K\) smallest \(F_i\)'s. The largest \(F\) within that set is the binding type.
\end{proof}

\subsection{Payoff: expected prize minus expected sanction}\label{subsec:payoff}
If a provider joins, the expected per-member payoff given belief \(\alpha\) is
\begin{align}
U_J(\alpha)
&=\ \underbrace{\pi_{n,K}(\alpha)\,\frac{1-p_K}{K}\,\omega}_{\text{expected prize}}
\;-\;
\underbrace{\bar p(\alpha)\,F_{\mathrm{eff}}(n,K)}_{\text{expected sanction}},
\label{eq:U_join}\\[3pt]
\bar p(\alpha)
&:=\ \tilde p\;+\;\pi_{n,K}(\alpha)\,\big(p_K-\tilde p\big).
\label{eq:barp}
\end{align}
Equivalently,
\begin{equation}
U_J(\alpha)\;=\;\underbrace{-\tilde p\,F_{\mathrm{eff}}(n,K)}_{\text{attempt cost}}
\;+\;
\pi_{n,K}(\alpha)\,\underbrace{\Big(\frac{1-p_K}{K}\,\omega-[p_K-\tilde p]\,F_{\mathrm{eff}}(n,K)\Big)}_{\text{success bonus}}.
\label{eq:attempt-bonus}
\end{equation}
Not joining yields \(0\).

\begin{remark}[Group Rationality]
The relevant slope in beliefs is positive whenever the expected marginal benefit 
from additional coordination outweighs the marginal increase in expected detection. 
Under independent detection, this reduces to an odds-ratio condition:
\begin{equation}
\frac{\omega}{K\,F_{\mathrm{eff}}}\;>\;\frac{q}{1-q}.
\label{ass:GC}
\end{equation}
Intuitively, the group’s per-capita prize–to–sanction ratio must exceed the 
individual detection odds.
\end{remark}

\begin{lemma}[Monotonicity in beliefs]\label{lem:monotone-alpha}
From equation ~\ref{ass:GC}, \(U_J(\alpha)\) is strictly increasing in \(\alpha\).
\end{lemma}

\begin{proof}
Since \(U_J(\alpha)= -\tilde p\,F_{\mathrm{eff}}
+\pi_{n,K}(\alpha)\Big(\tfrac{1-p_K}{K}\,\omega-(p_K-\tilde p)\,F_{\mathrm{eff}}\Big)\),
we have
\[
\frac{dU_J}{d\alpha}=\pi'_{n,K}(\alpha)\,
\Big(\tfrac{1-p_K}{K}\,\omega-(p_K-\tilde p)\,F_{\mathrm{eff}}\Big).
\]
Because \(\pi'_{n,K}(\alpha)>0\) for \(\alpha\in(0,1)\) and the bracket is positive via \ref{ass:GC}, \(dU_J/d\alpha>0\).
\end{proof}

\begin{proposition}[Multiple equilibria]\label{prop:corners}
For \(n\ge2\) and majority \(K\), the no-join profile is always a symmetric equilibrium. 
In addition, the all-join profile is a symmetric equilibrium if and only if
\begin{equation}
U_J(1)\;=\;\frac{1-p_K}{K}\,\omega\;-\;p_K\,F_{\mathrm{eff}}(n,K)\;\ge\;0.
\label{eq:corner-test}
\end{equation}
Under Assumption~\ref{ass:GC}, these are the only symmetric pure equilibria.
\end{proposition}

\begin{proof}[Proof sketch]
With \(\alpha=0\), \(\pi_{n,K}(0)=0\) and \(\bar p(0)=\tilde p\), so \(U_J(0)=-\tilde p\,F_{\mathrm{eff}}<0\); hence no-join is an equilibrium. 
With \(\alpha=1\), \(\pi_{n,K}(1)=1\) and \(\bar p(1)=p_K\), so all-join is an equilibrium iff \eqref{eq:corner-test} holds. 
Under Assumption~\ref{ass:GC} and Lemma~\ref{lem:monotone-alpha}, best responses are strictly increasing in \(\alpha\), which rules out additional symmetric pure equilibria between the corners.
\end{proof}

\subsection{Private information and selection}\label{subsec:gg}
\paragraph{Information structure.}
Let the fundamental \(\theta\in\mathbb{R}\) have continuous prior \(H\) with density \(h\) and full support. Each provider \(i\) observes
\[
s_i=\theta+\varepsilon_i,\qquad \varepsilon_i\ \text{i.i.d. with CDF }G_\sigma\ \text{and strictly log-concave density }g_\sigma,
\]
independent of \(\theta\). The parameter \(\sigma>0\) indexes information precision.

\begin{assumption}[Monotonicity in the fundamental]\label{ass:theta}
For each \(\alpha\), \(U_J(\alpha;\theta)\) is strictly increasing in \(\theta\).
\end{assumption}

\paragraph{Equilibrium.}
Let \(\alpha(\theta;\tau):=1-G_\sigma(\tau-\theta)\) denote the probability another provider joins when everyone uses cutoff \(\tau\).

\begin{proposition}[Unique symmetric cutoff equilibrium for each \(\sigma>0\)]
\label{prop:unique}
Under Assumption~\ref{ass:GC}, the information structure in \S\ref{subsec:gg}, 
and assuming \(U_J(\alpha;\theta)\) is strictly increasing in \(\theta\), 
there exists a unique symmetric Bayesian Nash equilibrium in monotone strategies: 
for each \(\sigma>0\) there is a unique threshold \(\tau_\sigma\) such that 
agent \(i\) joins iff \(s_i\ge\tau_\sigma\). The cutoff satisfies
\begin{equation}
\mathbb{E}\!\left[\,U_J\!\big(\pi_{n,K}(\alpha(\theta;\tau_\sigma));\theta\big)\ \big|\ s_i=\tau_\sigma\,\right]=0.
\label{eq:cutoff-characterization}
\end{equation}
\end{proposition}

\begin{proof}[]
Strict log-concavity implies MLRP, so best responses are threshold in own signal. Given a candidate cutoff \(\tau\), the LHS of \eqref{eq:cutoff-characterization} is continuous and strictly decreasing in \(\tau\): as \(\tau\) rises, \(\alpha(\theta;\tau)\) falls (reducing \(\pi_{n,K}\) and \(U_J\)), while conditioning on a higher \(s\) shifts \(\theta\) upward (Assumption~\ref{ass:theta}); the first effect dominates by Lemma~\ref{lem:monotone-alpha}. Boundary conditions ensure a unique zero. Uniqueness of the cutoff equilibrium then follows. See~\cite{MorrisShin2003} or ~\cite{FrankelMorrisPauzner2003} for closely related arguments.
\end{proof}

\begin{proposition}[Limit selection as \(\sigma\to0\)]
\label{prop:limit}
Let \(\theta^\star\) uniquely solve \(U_J(1;\theta^\star)=0\).
As \(\sigma\to0\), \(\tau_\sigma\to\theta^\star\),
and for \(\theta>\theta^\star\) (resp.\ \(\theta<\theta^\star\)) the probability of the collusive profile tends to one (resp.\ zero).
\end{proposition}

\begin{proof}[]
Under Assumptions~\ref{ass:theta} and \ref{ass:GC}, the complete-information game has exactly the two pure equilibria in Proposition~\ref{prop:corners}. The set of \(\theta\) for which all-join is an equilibrium is \(\{\theta:U_J(1;\theta)\ge0\}\), which has boundary \(\theta^\star\). Standard global-games arguments (e.g., ~\cite{CarlssonVanDamme1993,MorrisShin2003,FrankelMorrisPauzner2003}) imply the unique monotone equilibrium selects this boundary in the vanishing-noise limit.
\end{proof}

\subsection{Zipf heterogeneity and majority}\label{subsec:zipf}
Motivated by well-known rank--size regularities in firms \cite{axtell2001zipf}, we adopt a Zipf law with slope \(s=1\):
\begin{equation}
F_{(r)}\;=\;C\,r^{-1},\qquad r=1,\dots,n,
\label{eq:zipf}
\end{equation}
for scale parameter \(C>0\). Under majority \eqref{eq:majority},
\begin{equation}
F_{\mathrm{eff}}(n,K)\;=\;F_{(n-K+1)}\;=\;\frac{C}{n-K+1}
\;=\;
\begin{cases}
\displaystyle \frac{C}{K}, & n\ \text{odd},\\[6pt]
\displaystyle \frac{C}{K-1}, & n\ \text{even}.
\end{cases}
\label{eq:Feff-zipf}
\end{equation}
Let \(a:=1-q\in(0,1)\).
Substituting into \eqref{eq:corner-test} yields, for \(n\) odd,
\begin{equation}
U_J(1)\;=\;\frac{1}{K}\,\Big\{(\omega+C)\,a^{\,K}-C\Big\}.
\label{eq:corner-zipf}
\end{equation}
The bracket is strictly decreasing in \(K\), so the sign of \(U_J(1)\) is single-crossing in \(K\).

\begin{corollary}[Closed-form deterrence thresholds for odd $n$]
All-join is not an equilibrium iff $(\omega+C)(1-q)^K\le C$, equivalently
\[
K\ \ge\ \frac{\log\!\big(\frac{C}{\omega+C}\big)}{\log(1-q)}\quad\text{or}\quad
q\ \ge\ 1-\Big(\frac{C}{\omega+C}\Big)^{1/K}.
\]
\end{corollary}
\noindent These thresholds yield immediate policy comparatives: raising $q$ (detection), raising $C$ (sanctions), or raising $K$ (required majority size) each shrinks the collusive region.

\begin{remark}[Single-crossing in \(K\) (sign)]
For fixed parameters, the sign of \(U_J(1)\) as a function of \(K\) changes at most once, which is the relevant property for existence of the collusive corner.
\end{remark}

\subsection{Design bound}\label{subsec:design}
A conservative design sets \(V\) so that joining is unprofitable even if success is assured:
\begin{equation}
V^{\mathrm{safe}}(n,K)\;=\;\frac{K}{(1-p_K)\,\beta}\,p_K\,F_{\mathrm{eff}}(n,K).
\label{eq:Vsafe}
\end{equation}

\section{Calibration and Parameterization}\label{sec:calibration}

We calibrate the conservative design bound in Eq.~\eqref{eq:Vsafe}:
\[
V^{\mathrm{safe}}(n,K)
\;=\;\frac{K}{(1-p_K)\,\beta}\,p_K\,F_{\mathrm{eff}}(n,K),
\]
where \(F_{\mathrm{eff}}(n,K)\) is the binding per‑member sanction from Lemma~\ref{lem:binding}, \(p_K=1-(1-q)^K\) is the detection probability for a size‑\(K\) coalition within the effective exploitation window (\S\ref{subsec:detection}), and \(\beta\in(0,1]\) is the fraction of system stock \(V\) extractable as flow before suspicion, patching, or market responses curtail it (\S\ref{subsec:env}, Eq.~\eqref{eq:flow}). 

\paragraph{Sanction scale \(F_{\mathrm{eff}}\).}
In practice, \(F_{\mathrm{eff}}\) is driven by tangible financial and reputational losses borne by large providers:
\begin{enumerate}[label=(\roman*)]
  \item \emph{Equity drawdowns.} Event‑study evidence places the average post‑breach equity loss for public firms in the \(7\text{–}8\%\) range; we adopt \(7.5\%\) as a reference \cite{hbr2023cyberbreach}. Single‑day cloud‑related moves of \(\sim\$90\)B in market value have occurred at hyperscalers (e.g., Amazon) \cite{reuters2025amazon}, and sector examples (e.g., Okta) show double‑digit drops on breach news (about \(12\%\)) \cite{reuters2023okta}. For a top‑tier market cap \(M_{\text{top}}\!\approx\!\$1.8\)T, \(0.075\times M_{\text{top}}\) implies \(\$135\)B of equity loss on announcement.
  \item \emph{Legal and regulatory exposure.} GDPR empowers fines up to \(4\%\) of global annual turnover \cite{gdpr2016}; large U.S. cases have settled in the hundreds of millions (e.g., Equifax’s \(\$575\)M global settlement) \cite{ftc2019equifax}.
  \item \emph{Customer flight and financing costs.} Breaches damage trust and can worsen financing terms; major credit‑rating agencies increasingly incorporate cybersecurity posture into ratings and outlooks \cite{wapo2023creditcyber}.
  \item \emph{Remediation overhead.} Post‑incident audit, re‑platforming, and security spend are material and reduce margins; these are widely documented in breach case studies \cite{hbr2023cyberbreach}.
\end{enumerate}

\noindent\emph{Note} The cited magnitudes largely reflect ``innocent'' breaches (operational failures, compromised vendors). A malicious, profit‑motivated collusion to steal secrets would plausibly inflict a larger reputational penalty; in addition, criminal liability (including potential jail time) strengthens the effective penalty further.

\begin{figure}
    \centering
    \includegraphics[width=1\linewidth]{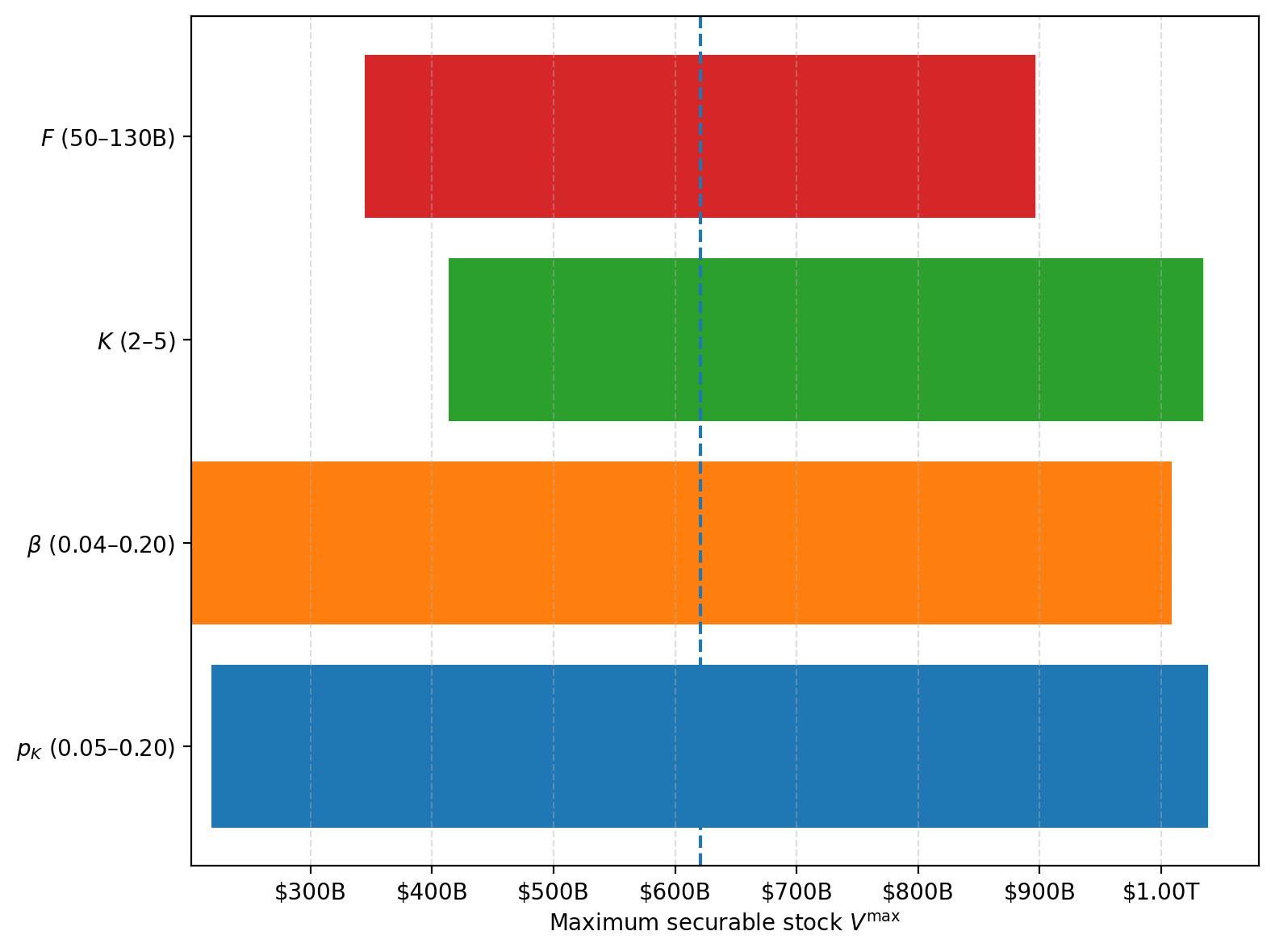}
    \caption{Tornado Graph of total effective security under variable ranges}
    \label{fig:tornado}
\end{figure}

\paragraph{Detection at the coalition \(p_K\).}
Bryant and Eckard (1991) estimate the risk of a cartel getting caught at 13\% to 17\% annually\cite{bryant1991price}. There is natural uncertainty on this number, since it can be hard to get data the collusions which were not detected. But on the flipside, the physical traces from TEE breaches, faulty attestation possibilities, etc. may make the probability uniquely higher in our setting. We take \(p_K=0.15\) and explore \(p_K\in[0.05,0.20]\). 
\[
q\;=\;1-(1-p_K)^{1/K}
\qquad\text{(e.g., }K{=}3,\ p_K{=}0.15\Rightarrow q\approx 5.3\%\text{ per window).}
\]

\paragraph{Flow fraction \(\beta\).}
We cap \(\beta\) using market turnover and operational controls. Annual equity turnover in large markets is often \(0.6\text{–}0.8\) \cite{fred_turnover_usa}, meanwhile 10-day median dwell times for many incidents (cf.\ \cite{mandiant2024mtrends}), along with but protocol‑level patch \& re‑attestation requirements can reduce this proportion further. There are endogenous considerations which limit effective extraction as well: 
\begin{itemize}
    \item the full value of the flow may not be directly extractable (e.g. only a proportion of it may be capturable via e.g. information arbitrage and front-running)
    \item market microstructure theory \cite{glosten1985bid,kyle1985continuous}
and practice\cite{cartea2023detecting} predicts that liquidity contracts when flow looks suspicious, which can dry up the exploitable flow.
\end{itemize}
We consider a baseline \(\beta=0.06\) as a conservative estimate and examine \(\beta\in[0.03,0.10]\).

\paragraph{Threshold.}
We illustrate with \(n=5\) and \(K=3\). With constant \(F_{\mathrm{eff}}\), changing \(K\) scales \(V^{\mathrm{safe}}\) linearly (holding \(p_K\) fixed).

\paragraph{Back‑of‑the‑envelope baseline.}
With \(K=3\), \(p_K=0.15\), \(\beta=0.06\), and \(F_{\mathrm{eff}}=\$135\)B,
\[
V^{\mathrm{safe}}
\;=\;\frac{3}{0.06}\cdot\frac{0.15}{0.85}\cdot \$135\text{B}
\;\approx\;\mathbf{\$1.19\ \text{trillion}}.
\]
A \emph{tighter} throttle (smaller \(\beta\)) raises the bound (e.g., \(\beta=0.055\Rightarrow\sim\$1.30\)T); a \emph{looser} throttle (\(\beta=0.065\)) lowers it (\(\sim\$1.10\)T).

\paragraph{Tornado bands.}
Unless stated otherwise we vary one parameter at a time around the baseline:
\[
F_{\mathrm{eff}}\in[\$100,\$135]\text{B},\quad
p_K\in[0.05,0.20],\quad
\beta\in[0.03,0.10],\quad
K\in\{3,4,5,6,7\}.
\]
Because \(V^{\mathrm{safe}}\) scales linearly in \(F_{\mathrm{eff}}\) and \(K\), with odds \(p_K/(1-p_K)\), and inversely in \(\beta\), the bars in Fig.~\ref{fig:tornado} are directly interpretable.

\begin{figure}
    \centering
    \includegraphics[width=0.75\linewidth]{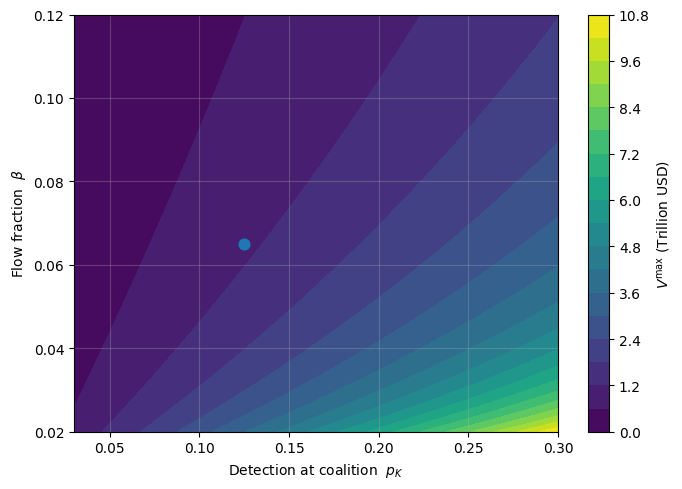}
    \caption{Iso Curves on Flow Level and Detection Probability}
    \label{fig:jointsensitivity}
\end{figure}

\section{Conclusion}
We derived a principal–agent model with a coalition threshold, layered detection, and heterogeneous sanctions yields tractable deterrence conditions and a conservative design bound $V^{\mathrm{safe}}$ that protocol designers can target ex ante. Under plausible parameters informed by time‑advantaged arbitrage, the cheapest collusion may remain uneconomical even at trillion‑dollar secured value.

The engineering of \name{} aligns with these levers: near‑stateless TEEs and periodic restarts compress extraction windows; DKG raises the threshold $K$; and accountability via attestations physical breach traces increase effective detection. 

Insights include:
\begin{itemize}
    \item Protocol-coordinated TEEs can achieve remarkable high levels of security, by effectively borrowing it from the reputational risk of existing firms.
    \item In an optimal design, the right unit of account for TEE compromise is \emph{flow}, not \emph{stock}.
    \item This can be achieved without relying on the individual security of a single TEE design.
\end{itemize}

We view improving empirical estimates of detection rates $q$ and extraction windows (domain‑specific $\beta$) as the key path to tighter calibrations and safer TEE‑augmented BFT systems.

\bibliographystyle{splncs04}
\bibliography{ref}

\end{document}